\documentclass[final,nomarks]{dmtcs-episciences}

\graphicspath{{graphics}}

\usepackage{xspace}
\usepackage{xcolor}
\usepackage{amsthm}
\usepackage{amsmath}
\usepackage{amssymb}
\usepackage{thmtools}
\usepackage[capitalize]{cleveref}

\usepackage[round]{natbib}
\usepackage[utf8]{inputenc}
\usepackage{subfigure}

\usepackage{doi}

\newtheorem{theorem}{Theorem}
\newtheorem{lemma}{Lemma}
\newtheorem{colorallary}{Corollary}
\newtheorem*{remark}{Remark}

\newcommand{\R}{\mathbb{R}}
\newcommand{\N}{\mathbb{N}}
\renewcommand{\P}{\mathbb{P}}
\newcommand{\eps}{\varepsilon}
\renewcommand{\L}{\mathcal{L}}

\newcommand{\col}{\phi}
\newcommand{\Q}{\ensuremath{\mathcal{Q}}\xspace}

\DeclareMathOperator*{\polylog}{polylog}

\newcommand{\etal}{et al.\xspace}

\author{Erwin Glazenburg\affiliationmark{1}\thanks{Supported by the Netherlands Organisation for Scientific Research (NWO) under project OCENW.M20.135.} \and Frank Staals\affiliationmark{1}}
\affiliation{Department of Information and Computing Sciences, Utrecht University, Utrecht, The Netherlands}

\title[On strictly output sensitive color frequency reporting]{On strictly output sensitive color frequency reporting\thanks{An earlier version of this article appeared at SOFSEM26: \cite{sofsemVersion}}}

\keywords{Data structure, Range searching, Lower bound, Frequency reporting}

\begin{document}

\publicationdata{vol. 28:4, SOFSEM 2026}{2026}{2}{10.46298/dmtcs.17715}{2026-03-13; 2026-03-13; 2026-06-23}{2026-06-25}

\maketitle

\begin{abstract}
  Given a set of $n$ colored points $P \subset \R^d$ we wish to store $P$ such that, given some query region $Q$, we can efficiently report the colors of the points appearing in the query region, along with their frequencies. This is the \emph{color frequency reporting} problem. We study the case where query regions $Q$ are axis-aligned boxes or dominance ranges. If $Q$ contains $k$ colors, the main goal is to achieve ``strictly output sensitive'' query time $O(f(n) + k)$. Firstly, we show that, for every $s \in \{ 2, \dots, n \}$, there exists a simple $O(ns\log_s n)$ size data structure for points in $\R^2$ that allows frequency reporting queries in $O(\log n + k\log_s n)$ time. Secondly, we give a lower bound for the weighted version of the problem in the arithmetic model of computation, proving that with $O(m)$ space one can not achieve query times better than $\Omega\left(\phi \frac{\log (n / \col)}{\log (m / n)}\right)$, where $\phi$ is the number of possible colors. This means that our data structure is near-optimal. We extend these results to higher dimensions as well. Thirdly, we present a transformation that allows us to reduce the space usage of the aforementioned data structure to $O(n(s \col)^\eps \log_s n)$. Finally, we give an $O(n^{1+\eps} + m \log n + K)$-time algorithm that can answer $m$ dominance queries in $\R^2$ with total output complexity $K$, while using only linear working space.
\end{abstract}

\section{Introduction}
Let $P$ be a set of $n$ points in $\R^d$, with $d = O(1)$, each of
which has a color in $\{1, \dots, \col\}$, and let $P_c$ denote the
subset of points of color $c$. We want to store $P$ so that given an
axis aligned query box $Q$ we can efficiently report the colors of the
points appearing in the query range, together with their
frequency. That is, we wish to report a set of $k$ color, frequency
pairs $(c,f)$ so that $f=|P_c \cap Q| > 0$ and the sum of the
frequencies is $|P \cap Q|$. See
Figure~\ref{fig:problemDescription}. We call this \emph{color
  frequency reporting}, but it is also known by the rather vague name
of \emph{type-2} color
counting \citep{gupta1995furtherResultsOnGeneralizedIntersectionSearching}. More
generally, we may associate every point $p \in P$ with a weight $w(p)$
from some appropriate semigroup, and instead report, for each color
$c$ appearing in the query range, the total weight $w(P_c \cap Q) = \sum_{p \in P_c \cap Q}w(p)$ in
the query range. Our main goal
is to obtain small, ideally (near-)linear-space data structures for
these problems that achieve a \emph{strictly output sensitive} query
time of the form $O(f(n) + k)$, i.e. a query time strictly linear in
the output size $k$ (and whose dependency $f(n)$ on $n$ is of course
also sublinear, and preferably even (poly-)logarithmic). We are
generally interested in data structures in the pointer machine model,
as this tends to give relatively simple and flexible data structures.

\begin{figure}[ht]
    \centering
    \includegraphics{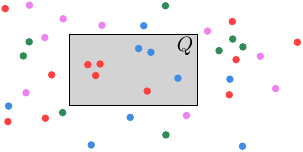}
    \caption{A set of colored points $P$ in $\R^2$ and a query box $Q$
      with color frequencies (red, 4) and (blue, 3). There are no green
      or purple points in $Q$, hence they do not appear in the output.
    }
    \label{fig:problemDescription}
\end{figure}

\subparagraph{Related work.} Color frequency reporting is a form of
\emph{generalized intersection searching} \citep{gupta1995furtherResultsOnGeneralizedIntersectionSearching,gupta2018ColoredSearchingSurvey,janardan1993generalizedIntersectionSearching}. For
points in $\R^1$, the problem was already solved optimally in 1995:
\cite{bozanis1995generalizedIntersectionSearching} gave
a data structure of linear size and query time $O(\log n +
k)$. Similarly, there is an optimal data structure for the dual
problem of reporting the frequencies of a set of colored intervals
stabbed by a query
point by \cite{gupta1995furtherResultsOnGeneralizedIntersectionSearching}.

In $\R^2$, it is relatively simple to obtain a data structure of size
$O(n \log n)$ with query time $O(\log n + k \log n)$: we first report
all intersected colors in $O(\log n + k)$ time using a colored range
reporting query (\cite{shi2005optimalColoredAlgs}), and then for all $k$
colors separately perform a standard uncolored range counting query in
$O(\log n)$ time each using range trees (\cite{Comp_geom_book}).

\cite{gupta2018ColoredSearchingSurvey} describe an
$O(n^{1+\eps})$ space data structure that, given a set of $n$ points
in $\R^d$, can report the $k$ distinct colors appearing in the query
range in $O(\log n + k)$ time. Here, and throughout the rest of the
paper $\eps > 0$ is an arbitrarily small constant. It is possible to
adapt their approach, which is based on bootstrapping, to support
frequency reporting queries as well. To achieve query time
$O(\polylog n + k)$, we need at least
$\Omega(n(\log n/\log\log n)^{d-1})$ space, as the lower bound
for range reporting by \cite{chazelle1990ReportingLowerBound} applies
(just assign each point a unique color), however, we are not aware of
any near-linear-space structures that achieve strictly output
sensitive query times. In particular, it is open whether one can
achieve even a query time of $O(n^{0.99}+k)$ using near-linear space.

In the more powerful word-RAM model, some recent progress has been
made towards near linear space structures for points in $\R^2$. \cite{chan2020furtherResultsOnColoredRangeSearching}
achieve query time $O\left(\frac{\log n}{\log\log n} + k \log\log n \right)$ using
$O(n \log^{1+\eps} n)$ space. They use a recursive grid approach and
several ``bit-tricks'', which allows them to sum the counts of multiple
colors simultaneously. \cite{afshani2023rangeSummaryQueries} present a linear-space data
structure for simplex color frequency reporting queries. They achieve
query time $O(n^{1-1/d}+n^{1-1/d}\phi^{1/d}/w^\alpha)$, where
$\alpha > 0$ is some constant and $w$ is the word size (which is
typically $\Theta(\log n)$). Furthermore, they solve approximate color
frequency counting with an \emph{additive} error $\eps n$ for
dominance queries with $O(n)$ space and $O\left(\log n + \eps^{-1}n + k \right)$
query time.

In external memory with block size $B$, \cite{ganguly2019categorical} give a data structure of size $O(n)$
that answers \emph{constant-factor approximate} frequency reporting
queries in $O(\log_B n + \log^* B + k/B)$ time (here,
  $\log^* B$ is the iterated logarithm of $B$, i.e. the smallest value
  $i$ for which $\log^{(i)} B = \log (\log^{(i-1)} B) \leq 2$ (and using
  that $\log^{(1)} B = \log B$)). In the pointer machine model, their techniques can be used to answer $(1+\eps)$-approximate frequency reporting queries using $O(n/\eps)$ space in $O(\log n + k)$ time.

\subparagraph{Challenges.} Achieving our goal of strict output
sensitivity using near linear space turns out to be extremely
challenging. Most data structuring approaches (including the
techniques used in the above results) follow a divide-and-conquer
approach that partitions the space (and thereby the input point
set). However, this may cause the points $P_c \cap Q$ of color $c$ to
be stored in many, say $g(n)$, different places in the data
structure. This means that we have to aggregate their frequencies at
query time. If this happens for $\Omega(k)$ colors, we get an
$\Omega(kg(n))$ cost in the query time. If we instead store the points
per color, we may have to spend $\Theta(\log n)$ time per color to
actually count the subset that lies in $Q$. Filtering
search \citep{chazelle1986filtering} is a common technique used in
reporting queries which allows us to charge some of the ``searching''
costs to the output. However, as our main challenge is combining the
color counts, even filtering search is of limited use.

\subparagraph{Results \& Organization.} We present the following
results:

\begin{itemize}
\item When $P$ is a set of $n$ points in $\R^2$, and we are given a
  parameter $s \in \{2\dots n\}$, we present a simple $O(ns\log_s n)$ size
  data structure storing $P$ that allows frequency reporting queries
  for dominance ranges in $O(\log n + k\log_s n)$ time. See Section~\ref{sec:An_abstract_data_structure}. For ease of
  exposition, we present our data structure for color frequency
  reporting, but our results apply also to the weighted version of the
  problem. Setting
  $s=n^\eps$, this then yields an $O(n^{1+\eps})$ space structure with
  strictly output sensitive $O(\log n + k)$ query time. This matches,
  and we believe simplifies, the result of \cite{gupta2018ColoredSearchingSurvey}. With this structure we can essentially interpolate between the $O(n^{1+\eps})$ space data structure with $O(\log n + k)$ query time, and the $O(n \log n)$ space data structure with $O(\log n + k \log n)$ query time. Setting e.g.
  $s=2^{\sqrt{\log n}}$ reduces the space to $O(n2^{\sqrt{\log n}})$
  space, while still answering queries at a cost of $\sqrt{\log n}$
  per reported color.

\item We extend the above results to queries with arbitrary axis
  aligned boxes in $\R^d$. Every additional dimension comes at the
  cost of an $O(s\log_s n)$ factor in space, and an $O(\log_s n)$ factor per
  color (and an additive $O(\log^{d-1} n)$ term) in query time. Every
  additional side we allow in the query range (from $d$-sided
  dominance ranges to arbitrary $2d$-sided axis-aligned boxes) comes
  at the cost of a factor $\log n$ in space, and a factor $2$ in query
  time. With $s \approx n^\eps$, this still leads to strictly
  output sensitive query time $O(\log n + k)$ using $O(n^{1+\eps})$
  space.

\item We prove a lower bound for the weighted case with dominance
  queries in the arithmetic model; see Section~\ref{sec:Lower_bound}. In this model, computation is essentially free: we
  count only the number of pre-computed weights (of subsets of points)
  stored by the data structure, and the number of these weights we
  have to combine to obtain the answer to a query. Building on
  the lower bounds for counting queries in this
  model by \cite{chazelle1990ArithmeticLowerBound}, we prove that using
  $O(m)$ space one can not achieve query time faster than
  $\Omega\left(\phi \left( \frac{\log (n / \col)}{\log(m/n)} \right)^
    {d-1}\right)$. In particular, compare this to our tradeoff
  data structure from Section~\ref{sec:An_abstract_data_structure},
  which uses $O(ns^{d-1} \polylog n)$ space for
  $O(\polylog n + k \log_s^{d-1} n)$ query time; the lower bound says
  that using $O(ns^{d-1})$ space, one cannot achieve query time better
  than
  $\Omega\left(\phi \left( \frac{\log (n/\col)}{\log((ns^{d-1})/n)} \right)^
    {d-1}\right) = \Omega\left(\phi \left( \frac{\log
        (n/\col)}{\log(s^{d-1})} \right)^ {d-1}\right) = \Omega\left(\phi
    \left( \frac{\log (n/\col)}{\log s} \right)^ {d-1}\right) = \Omega(\phi
  \log_s^{d-1} (n/\col))$. So for $\col \leq n^c$ with $c < 1$, our data
  structure is worst-case optimal with respect to the query dependency
  on $k$, up to $\log n$ factors in space.

  This suggests that it will be hard to get strictly output-sensitive query
  times for (weighted) color frequency reporting queries using near
  linear space. At the very least, one has to use properties specific
  to counting (e.g. that we can subtract frequencies) to achieve this
  goal. Note however that even in the uncolored case no such results
  are known.

\item We present a transformation that, for some values of $s$ and $\col$, allows us to reduce the
  space for 2D dominance queries to $O(n(s\col)^\eps \log_s n)$, by answering
  several one dimensional dominance queries simultaneously; see
  Section~\ref{sec:transformation}. Hence, this allows us to shave
  some logarithmic factor in space in certain scenarios.

\item Lastly, we present a result on the algorithmic version of the problem in which we are
  given a set of $n$ colored points $P \subset \R^2$ as well as a
  set \Q of $m$ 3-sided color frequency reporting queries that we
  \emph{all} have to answer. See
  Figure~\ref{fig:problemDescriptionOffline}. Our algorithm in
  Section~\ref{sec:A_linear_space_algorithm} uses linear,
  i.e. $O(n + m)$, working space, and answers all $m$ queries in
  $O(n^{1+\eps}+m\log n + K)$ total time, where $K$ is the total output
  complexity of all queries. Note that e.g. for $m = \Theta(n)$ we
  thus essentially obtain $O(n^{\eps}+k)$ time per query.

\end{itemize}

\begin{figure}[hb]
    \centering
    \includegraphics{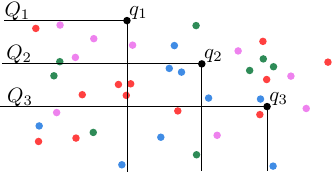}
    \caption{A set of points and three dominance queries. In the offline version, we would want to answer all three queries at once.}
    \label{fig:problemDescriptionOffline}
\end{figure}

\section{Preliminaries}
\label{sec:Preliminaries}

We focus on range searching where the query region
$Q = [a_1,b_1] \times [a_2,b_2] \times \dots \times [a_d,b_d]$ is an
axis-aligned box in $d$-dimensional space. A box in $\R^1$ is an
interval, and a box in $\R^2$ is a rectangle. In the simplest case,
the box is unbounded in one direction for every dimension, e.g. towards $-\infty$. A
query $Q=(-\infty,x] \times (-\infty,y] \times (-\infty,z] \times \dots$ of
this type is a \emph{dominance} query; see \cref{fig:problemDescriptionOffline}. We will often identify a dominance query $Q$ with its corner point $q=(x,y,z,\dots)$. A dominance query can be seen as a \emph{$d$-sided box}: it is bounded
on one side in each of the $d$ dimensions. Similarly an $i$-sided box,
for $d \leq i \leq 2d$, is bounded on two sides in $i-d$ dimensions,
and on one side in the remaining $2d - i$ dimensions. See \cref{fig:iSidedBoxes}.

\begin{figure}[ht]
    \centering
    \includegraphics{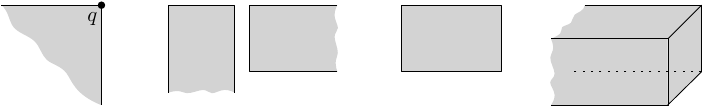}
    \caption{From left to right: $2$-, $3$-, and $4$-sided boxes in $\R^2$, and a $5$-sided box in $\R^3$.}
    \label{fig:iSidedBoxes}
\end{figure}

\subparagraph{Color frequency reporting queries in $\R^1$.} We briefly
review the approach of \cite{bozanis1995generalizedIntersectionSearching} (which itself
uses ideas from \cite{gupta1995furtherResultsOnGeneralizedIntersectionSearching})
that solves frequency reporting queries in $\R^1$. We focus on the case of dominance queries
$(-\infty,q]$. Fix a color $c \in \{1,\dots,\phi \}$, and let $x_i$ be the $i^\textrm{th}$ smallest
point in $P_c$. They then map $x_i$ to the point $(x_i,x_{i+1})$ in
$\R^2$ (or $(x_i,\infty)$ if no successor exists). Let $\tilde{P}_c$
denote the resulting set of mapped points, and let
$\tilde{P}=\bigcup_c \tilde{P}_c$. See \cref{fig:1Dstruct}. Bozanis \etal observe that the
query quadrant $\tilde{Q} = (-\infty,q] \times [q,\infty)$ contains at
most one point $(x_i,x_{i+1}) \in \tilde{P}$ of each color $c$, namely the
rightmost point of color $c$ in $(-\infty,q]$. Thus if for color $c$ query $\tilde{Q}$ contains point $x_i$, then the original query $Q$ contains exactly $i$ points of color $c$. They store the set $\tilde{P}$ in a priority search tree in order to quickly report (the ranks of) all $k$ points contained in query $\tilde{Q}$ in $O(\log n + k)$ time. 

We can answer weighted queries as well simply by setting the
weight of the point $(x_i,x_{i+1})$ to be $\sum_{h \leq i} x_h$. For two-sided
query intervals $[\ell,r]$ one can similarly report the leftmost
points of each color $c$ in the query range, and then instead report
the difference in ranks as the final count for color $c$; in this case the queries $\tilde{Q}$ will be three-sided. In the
weighted case this requires the weights to be taken from a group
(rather than a semigroup).

\begin{figure}[ht]
    \centering
    \includegraphics{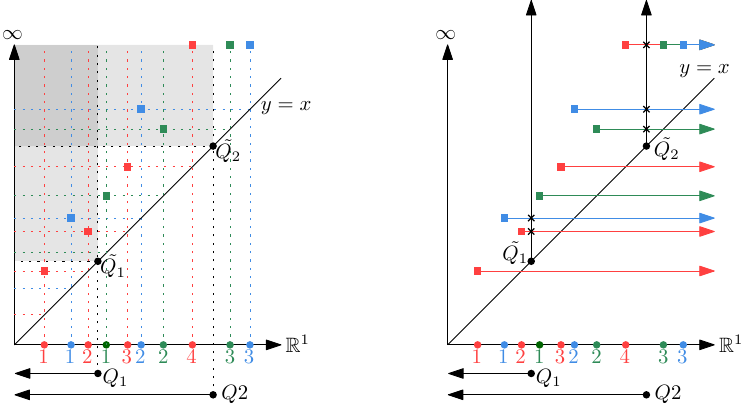}
    \caption{Left: an illustration of the 1D data structure from \cite{bozanis1995generalizedIntersectionSearching}. A set of points in $\R^1$ and two queries (on the $x$-axis) mapped to their corresponding counterparts in $\R^2$. Note that a mapped query corresponds to a quadrant originating from the line $y=x$, and that each query indeed contains at most one point of each color. Right: the same data structure implemented with halflines.}
    \label{fig:1Dstruct}
\end{figure}

\section{A strictly output sensitive data structure}
\label{sec:An_abstract_data_structure}

We first describe a generic data structure for dominance queries in
$\R^2$. We then extend the result to dominance queries in $\R^d$, and finally to arbitrary axis-aligned query boxes.

We first show that the 1D data structure by \cite{bozanis1995generalizedIntersectionSearching} can be implemented using a binary search rather than a priority search tree (which will allow us to use fractional cascading \citep{chazelle1986fractionalCascading} to speed up multiple 1D queries on different data structures with the same query).

\begin{lemma}
  \label{lem:1d_dominance_fc}
  Let $P$ be a set of $n$ points in $\R^1$ with colors from
  $\{1 \dots \phi \}$. In $O(n\log n)$ time we can build an $O(n)$ space data
  structure, so that we can answer dominance color frequency reporting queries using a binary search and $O(k)$ additional time.
\end{lemma}

\begin{proof}
  We implement the structure by \cite{bozanis1995generalizedIntersectionSearching}
  using a hyve graph \citep{chazelle1986filtering}. That is, Bozanis \etal
  transform every point $p \in P$ to a point $\tilde{p}$ in $\R^2$, so
  that one can answer a frequency counting query by reporting all
  points (one per color) from $\tilde{P} = \{\tilde{p} \mid p \in P\}$
  in a quadrant $\tilde{Q} = (-\infty,q] \times [q,\infty)$. We now
  map every point $\tilde{p}$ into a rightward ray starting at
  $\tilde{p}$, and store them in Chazelle's hyve graph \cite[Corollary
  1]{chazelle1986filtering}. We can then report all points in
  $\tilde{Q}$ by reporting all halflines that intersect the upward
  vertical ray originating from $(q,q)$. See the right side of \cref{fig:1Dstruct}. In particular, by doing a binary
  search for the topmost segment at $q$, and then walking vertically
  down, reporting the $k$ segments intersected by the ray until we
  reach point $(q,q)$. The construction time remains $O(n\log n)$. 
\end{proof}
Note that this technique does not extend to two sided query intervals $[\ell,r]$, since in that case we need to perform three-sided queries in $\R^2$ rather than dominance queries.

We now use this structure to answer dominance queries for points in $\R^2$.

\begin{lemma}
  \label{lem:2d_dominance}
  Let $P$ be a set of $n$ points in $\R^2$, each of which has a color from $\{1, \dots, \col\}$, and let $2 \leq s \leq n$ be a parameter. In
  $O(ns \log n\log_s n)$ time we can construct a data structure of
  size $O(ns\log_s n)$ that can answer dominance frequency reporting queries in
  $O(\log n + k\log_s n)$ time.
\end{lemma}

\begin{proof}
  We build a balanced $s$-ary tree on the $x$-coordinates of the
  points, of height $O(\log_s n)$. See \cref{fig:datastructTree}. Every node $\nu$ corresponds to a vertical strip $S_\nu$ and
  the subset of points $P_\nu = P \cap S_\nu$ in $S_\nu$; the root node $r$ corresponds to a singular 'strip' covering the whole plane, and $P_r = P$. Each strip $S_\nu$ is then horizontally partitioned into $s$ vertical strips $S_1,\dots,S_s$, one strip $S_i$
  for each child $c_i$ of $\nu$, each of which contains a $1/s$ fraction of the points.

  Consider a node $\nu$. Let $L_i = \bigcup_{h < i} P_{c_h}$ be the subset of points in $P_\nu$ that lie left of $S_i$,
  and for a point $p \in L_i$ let $p'$ be the $y$-coordinate of $p$, and let $L'_i = \{p' \mid p \in L_i\}$ denote the
  resulting set of 1D points ($y$-coordinates). We essentially regard $p'$ as the projection of $p$
  onto the left boundary of strip $S_i$, see
  \cref{fig:largeSpaceExplanation}. We store each set $L'_i$ in a
  one-dimensional color frequency reporting data structure
  $\mathcal{L}_i$ as described in \cref{lem:1d_dominance_fc}. We do this for every node $\nu$ and every strip $S_i$ in $\nu$.

  \begin{figure}[ht]
    \centering
    \includegraphics{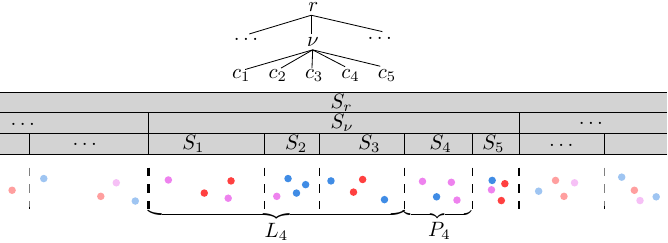}
    \caption{A set of points in $\R^2$ with an $s$-ary strip-tree on its $x$-coordinates (with $s=5$). The points $P_4$ and $L_4$ for strip $S_4$ are indicated.}
    \label{fig:datastructTree}
  \end{figure}
  
  \begin{figure}[ht]
    \centering
    \includegraphics{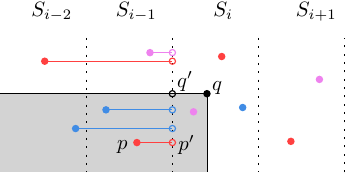}
    \caption{A dominance query whose corner point $q=(x,y)$ lies in
      strip $S_i$. }
    \label{fig:largeSpaceExplanation}
  \end{figure}

  We use fractional cascading \citep{chazelle1986fractionalCascading} to link all of these $\mathcal{L}_i$
  data structures in the entire tree together, so that given a query
  value $y$ we can perform the binary searches for $y$ in all
  $O(\log_s n)$ nodes along a root-to-leaf path in $O(\log n)$ time in
  total, rather than $O(\log_s n \log n)$ time. To this end, we actually implement every node $\nu$ as a
  small binary tree on its $s$ children (so that the degree of each
  node is constant).

  The total space usage is $O(ns)$ per level, and thus $O(ns\log_s n)$
  in total. Furthermore, constructing all these data structures takes
  $O(ns\log n)$ time per level, and thus $O(ns\log n\log_s n)$ time in
  total.

  \subparagraph{The query algorithm.} To answer a query
  $Q=(-\infty,x] \times (-\infty,y]$ we follow a root-to-leaf path towards
  $x$ while maintaining the frequency of each (relevant) color. At
  node $\nu$, let $S_i$ be the (sub)strip of $S_\nu$ containing $x$. We then query $\mathcal{L}_i$ with $(-\infty,y]$ and
  update the frequencies of the at most $k$ reported colors. Indeed
  observe that (i) a point $p \in L_i$ left of strip $S_i$ appears in $Q$
  if and only if $p' \in (-\infty,y]$, (ii) that none of the points right of
  $S_i$ lies in $Q$, and (iii) that the points in $S_i$ are accounted for in
  the subtree corresponding to $S_i$.

  At each node of the search path we have to update the frequencies
  for at most $k$ reported colors. Note that this is not as trivial as
  it might sound, especially in the pointer machine model, as the
  data structures of Lemma~\ref{lem:1d_dominance_fc} do not report the colors
  in a fixed order. Hence, we may have to sort these colors in
  $O(k\log k)$ time before merging the two lists of color counts. To
  update the color frequencies in $O(k)$ time per node, we instead use
  a mechanism similar to \cite{gupta1997rangeRestrictions}. Our data structure also
  stores a linked list of length $\phi$ whose $c^\textrm{th}$ node
  $\mu_c$ (with $c \in \{ 1, 2, \dots \col \}$) maintains the current
  total count of color $c$. If a query on one of the subsets returns
  some tuple $(c, k_c)$ in an answer, we increment the value at $\mu_c$
  by $k_c$. For each point of color $c$ stored in one of the
  Lemma~\ref{lem:1d_dominance_fc} data structures, we also store a
  pointer to node $\mu_c$. This allows us to update the answer in $O(1)$
  time per color, at no extra space cost. We additionally maintain a
  list of pointers to all non-zero nodes to report all final counts in
  $O(k)$ time (whenever we increment the count of a color $c$ and the
  previous count of that color was $0$, we add a pointer to $\mu_c$ to
  this list). We appropriately reset these linked lists after
  answering a query.
  
  \subparagraph{Query time.} We can find the root-to-leaf path towards $x$ in $O(\log n)$ total time by doing a single binary search on the leaves, and walking up towards the root. This path has length $O(\log_s n)$, and at each node we perform one 1D query in $O(\log n + k)$ time; note that each of the $k$ output colors may appear in every node. Using the above procedure of updating frequencies, the total query time is thus $O(\log n + k \log_s n)$ (recall that the 1D structures are linked using fractional cascading). 
\end{proof}

\subsection{Dominance queries in higher dimensions.} We extend the
above approach to the case of dominance queries in $\R^d$, for
$d \geq 2$. Again, this result applies to the weighted case as wel.

\begin{theorem}
  \label{thm:tradeoff_dominance_data_structure}
  Let $P$ be a set of $n$ points in $\R^d$, with $d \geq 2$, each of
  which has a color from $\{1, \dots, \col\}$, and let $s \in \{ 2, \dots, n \}$ be
  a parameter. In $O(n\log n(s\log_s n)^{d-1})$ time, we can build a
  data structure of size $O(n(s \log_s n)^{d-1})$ that can answer a
  frequency reporting dominance query $Q$ in
  $O(\log n \log_s^{d-2} n + k\log_s^{d-1} n)$ time.
\end{theorem}
\begin{proof}
  We prove this by induction on $d$. For $d = 2$, we simply build the
  data structure of Lemma~\ref{lem:2d_dominance}, which has the
  desired space usage, and query and preprocessing time.

  For the inductive step we first describe the data structure and
  query algorithm, and then analyze the space usage and query time.

  \subparagraph{The data structure and the query algorithm.} As in
  Lemma~\ref{lem:2d_dominance} we build a balanced $s$-ary tree on
  the $x$-coordinates of the points in $P$. Let $S_1,\dots,S_s$ again
  denote the vertical strips (whose bounding planes have a normal in
  the direction of the $x$-axis) that partition the strip of a node
  $\nu$. Define $P_i$ and $L_i$ as before. For each point $p \in L_i$
  we again consider the $(d-1)$-dimensional point $p'$ obtained by
  discarding the $x$-coordinate of $p$, and similarly let
  $P' = \{p' \mid p \in P\}$. For each strip $S_i$ we construct the
  $(d-1)$-dimensional data structure $\L_i$ on the points $L_i'$.

  To answer a dominance query $Q=(-\infty,q_x] \times Q'$ we again
  follow the root-to-leaf path towards $q_x$, query
  the appropriate associated structures $\L_i$ at each node, and keep track of the color frequencies as before.

  \subparagraph{Space usage and preprocessing time.} Consider the
  space used by the root node of the tree. Each $(d-1)$-dimensional
  data structure stored there takes $O(n (s \log_s n)^{d-2})$
  space. We build $s$ of them for a total space usage of
  $O(s n (s \log_s n)^{d-2})$. Every point is involved in only one
  recursive subproblem, so the total number of points on each level of
  the recursion tree is $n$, and thus on each level we use
  $O(s n (s \log_s n)^{d-2})$ space. The total
  space usage is thus
  $O((\log_s n) s n (s \log_s n)^{d-2}) = O(n (s \log_s n)^{d-1})$, as
  claimed.

  Analogously, the total preprocessing time is
  $O(n\log n(s\log_s n)^{d-1})$.

  \subparagraph{Query time.} We find the root-to-leaf path in $O(\log n)$ time. For each node on this path we query the appropriate $(d-1)$-dimensional data
  structure in $O(\log n \log_s^{d-3} n + k\log_s^{d-2} n)$ time. Since the tree has height $O(\log_s n)$, the total query time is $O(\log_s n(\log n \log_s^{d-3} n + k\log^{d-2} n)) = O(\log n \log_s^{d-2} n +
  k\log_s^{d-1} n)$.
\end{proof}

\subsection{From dominance queries to general box queries}

We now consider answering queries for more general types of boxes. 

\begin{restatable}{lemma}{twosidedFourSidedReduction}
\label{lem:data_structureIncreaseNumberOfSides}
Assume we are given a data structure that can report color frequencies with $i$-sided box queries (with $d \leq i \leq 2d$), using $S(n) \geq O(n)$ space with $Q(n,k) \geq O(\log n)$ query time. We can then create a data structure that can answer $j$-sided box queries (with $i \leq j \leq 2d$ using $O(S(n) \log^{j-i}n)$ space with $O(2^{j-i} Q(n,k))$ query time.
\end{restatable}

\begin{proof}
We first show how to answer $(i+1)$-sided box queries using the given data structure; we can answer $j$-sided box queries by repeating this procedure $j-i$ times. We assume the $(i+1)$-sided box is bounded from two sides in the $x$-dimension; by reordering the coordinates we can answer any type of $(i+1)$-sided query.

We build a balanced binary tree $T$ on the $x$-coordinates of the points $P$. For a node $u$, let $P_u$ be the set of points contained in the subtree of $u$, and let $x_u$ be the median $x$-coordinate in $P_u$ (the $x$-coordinate that splits $u$'s children).

Let $Q$ be an $(i+1)$-sided query with $x$-range $[x_1,x_2]$ that we wish to answer. We
find the highest node $u \in T$ such that $x_u \in [x_1,x_2]$; let $c_1$ and
$c_2$ be its children. Note that $Q$ is split in two by $x_u$. See
Figure~\ref{fig:dominanceToRectangle}. Consider the two $i$-sided
queries $Q_1$ and $Q_2$ that are equivalent to $Q$ except for their
$x$-range, where $Q_1$ has range $[x_1,\infty)$ and $Q_2$ has range
$(-\infty,x_2]$. Now it holds that $Q \cap P_u = (Q_1 \cap P_{c_1}) \cup
(Q_2 \cap P_{c_2})$, i.e. the set of points from $P_u$ contained in
$Q$ is the same as the set of points from its left child contained in
$Q_1$ (the red area in the figure) combined with the set of points
from its right child contained in $Q_2$ (the green area); note that this is a disjoint union. Therefore,
by answering these two $i$-sided queries and combining the answers we
can answer the $(i+1)$-sided query $Q$. So, we store an $i$-sided data structure in each node of the tree, using $O(S(n)\log n))$ total space, answering queries in $O(\log n + 2Q(n)) = O(2Q(n))$ time.

Each of the $j-i$ steps we add one side to the query box, and we gain a factor $\log n$ in space and a factor $2$ in query time. The lemma follows.
\end{proof}

 \begin{figure}[ht]
    \centering
    \includegraphics{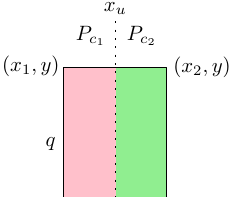}
    \caption{A 3-sided query $q$ at node $u$, with $x_u$ shown. }
    \label{fig:dominanceToRectangle}
  \end{figure}

\begin{remark}
One might think inclusion/exclusion offers a simpler solution. Consider a 2D rectangle query $Q$ with corners $tl,tr,bl,br$ (topright, topright, bottomleft, bottomright). We can perform a dominance query on all four cornerpoints and combine the resulting counts using inclusion/exclusion ($tr - tl - br + bl$) to obtain the answer to the rectangle query. However, suppose that rectangle query $Q$ contains only very few colors $k \ll \col$, but the dominance query from e.g. $tr$ contains all $\col$ colors. Answering dominance query $tr$ would take at least $\Omega(\col)$ time, so this technique is not strictly output sensitive, and thus not applicable for our use-case.
\end{remark}

Combining \cref{thm:tradeoff_dominance_data_structure,lem:data_structureIncreaseNumberOfSides} we thus get the following result for general $(2d)$-sided query boxes:

\begin{theorem}
  \label{col:tradeoff_datastructure_general_boxes}
  Let $P$ be a set of $n$ points in $\R^d$, each of which has a color from $\{1, \dots, \col\}$, and let $2 \leq s \leq n$ be a parameter. In
  $O(n \log^{d+1} n (s \log_s n )^{d-1})$ time, we can build a data structure of size
  $O(n \log^d n (s \log_s n)^{d-1})$ that can answer axis-aligned frequency reporting queries in
  $O(\log n \log^{d-2}_s n + k \log_s^{d-1} n)$ time.
\end{theorem}

This result again also applies to the weighted case.

Our main goal was to achieve strictly output sensitive query time. We
can achieve this by using
\cref{col:tradeoff_datastructure_general_boxes} with $s = n^{\eps/d}$,
for any constant $\eps > 0$, resulting in a data structure with
$O(n \log^d n (n^{\eps/d} d/\eps)^{d-1}) \leq O(n^{1+\eps})$ space,
$O(n \log^{d+1} n (n^{\eps/d} d/\eps)^{d-1}) \leq O(n^{1+\eps})$
preprocessing time, and
$O(\log n (d/\eps)^{d-2} + k (d/\eps)^{d-1}) = O(\log n + k)$ query time (using
that $\log_{n^{(\eps/d)}}(n) = d/\eps = O(1)$). (In this version, we
can even forgo the fractional cascading approach, and use the original
data structure by \cite{bozanis1995generalizedIntersectionSearching} to solve the
queries in $\R^1$, since the recursion tree has constant height.) The overall result matches that of \citet[Corollary 4]{gupta1997rangeRestrictions}, however we believe that
our approach is simpler.

\section{A lower bound in the arithmetic model}
\label{sec:Lower_bound}
In this section we focus on the weighted version of the problem, where
each point in $P \subset \R^d$ also has a weight from some
semigroup $(S,+)$. Given a (dominance) query $Q$ we now wish to report, for
every color $c$ appearing in $Q$, the total weight $w(P_c \cap Q)$.  

Our lower bound holds in the arithmetic model, and builds on a result
by \cite{chazelle1990ArithmeticLowerBound} for the
uncolored version of this problem. We describe the model here briefly,
see Section 2 of Chazelle's article for more detail. We then state
Chazelle's result, and afterwards show how to apply it to obtain a
lower bound for our colored problem.

\subparagraph{The arithmetic model.}  In the arithmetic model we can store
\emph{generators}, which consist of a subset of points, and the total
sum of their weights. A collection of generators is called a
\emph{storage scheme $\Gamma$} if, for any query $Q$, we can find a
subset of generators containing exactly the points in $P \cap Q$. The
answer to an (uncolored) query is then simply the sum of the values of these
generators. The size of $\Gamma$ is the number of generators stored,
and the cost of answering a query is the number minimum number of
generators required to compute the answer. Importantly, the query cost
does not include the time required to actually find the correct generator(s), nor the size of these generators,
and as such the arithmetic model is very suitable for proving lower
bounds.

As a simple example, consider the 1D colored weight reporting problem
in the arithmetic model. The data structure from
Lemma~\ref{lem:1d_dominance_fc} stores $n$ generators: each point $p$
of color $c$ stores the total weight of all points of color $c$ left
of $p$. For any query, the correct weight for each color is thus
stored in a single generator. In the arithmetic model, the query cost
is thus simply $k$, while in the pointer machine model we still need
extra techniques (a transformation to 2D and a priority search tree)
to find the correct generators in $O(\log n)$ time.

We require that the used semi-group $(\mathcal{S},+)$ is \emph{faithful}, which is the case  if for each $n > 0$, $\emptyset \subset T_1, T_2 \subseteq [1 \dots n]$, $T_1 \neq T_2$, and every sequence of integers $\alpha_i, \beta_j > 0$ ($i \in T_1, j \in T_2$), the equation $\sum_{i \in T_1} \alpha_i s_i = \sum_{j \in T_2} \alpha_j s_j$ is not an identity (i.e., cannot be satisfied for all assignments of the variables $s_1, \dots , s_n \in S$) (definition from \citet[Section~2.3]{chazelle1990ArithmeticLowerBound}).
Most semigroups arising in practice, e.g. $(\N,+)$, $(\N,\max)$, $(\{0,1\}, \text{or})$, have this property. Essentially, this means that a storage scheme must work for \emph{any} weight assignment: if we change the weight assignment, then the storage scheme must still work without changing the subset of points contained in a generator (but of course the total sum of the weights of the points in a generator may change).

\subparagraph{The uncolored lower bound.}  For a point set $P$ in
$\R^d$ of a single color, a storage scheme $\Gamma$ for $P$, and a dominance query
$Q$, let $t(P,\Gamma,Q)$ be the cost of answering query $Q$, i.e. the
smallest number of generators from $\Gamma$ needed to answer
$Q$. Chazelle shows the following:

\begin{lemma}[Chazelle's uncolored lower bound (\citeauthor{chazelle1990ArithmeticLowerBound}, \citeyear{chazelle1990ArithmeticLowerBound}, Theorem~2.1) ]
\label{lem:chazellesLemma}
Let $\eps > 0$ be a constant, let $P$ be a random set of $n$
points in $\R^d$ (each drawn randomly and independently from a uniform distribution on $[0,1]^d$), each with a weight from some faithful semigroup, let
$\Gamma$ be any storage scheme for $P$ of size $m$, and let $Q$ be a
random dominance query. There exists a constant $r$ (depending on
$\eps$) such that, with probability $(1-\eps)$, it holds that
$t(P,\Gamma,Q) \geq r \left( \frac{\log n}{\log (2m/n)}
\right)^{d-1}$.
\end{lemma}

Essentially, there is a high probability that a query $Q$ is
``difficult'' to answer, meaning that we have to add together at least
$r \left( \frac{\log n}{\log (2m/n)} \right) ^{d-1}$ generators.

\subparagraph{The colored lower bound.} We now show how to extend this
lower bound to the colored version of the problem. The input point set
now contains points of multiple colors, and thus in principle we could
have generators consisting of points of multiple colors too. We first
show that a storage scheme can never use such multi-colored
generators, using the assumption that the semi-group is faithful.

\begin{lemma}
  Let $P$ be a random set of $n$ colored points, each with a weight
  from some faithful semigroup, and let $\Gamma$ be any optimal
  storage scheme for $P$. Each generator in $\Gamma$ contains points
  of only a single color.
\end{lemma}
\begin{proof}
Recall that faithfulness implies that a storage scheme must work for any assignment of weights to the points. Assume that there exists some storage scheme $\Gamma$ that, to compute the weight of some color $c_1$ in some query $Q$, uses a generator $\gamma$ that contains (among others) a point $p$ of a different color $c_2$. Then if we change the weight of point $p$, e.g. by doubling its weight, the total weight of color $c_1$ computed by storage scheme $\Gamma$ also changes, even though the actual weight of all points of color $c_1$ did not change. This means that storage scheme $\Gamma$ does not work for all weight assignments, and is thus not a proper storage scheme; a contradiction. We conclude that, to compute the weight of some color $c_1$ in a query $Q$, a storage scheme may only use generators that exclusively contain points of color $c_1$. The lemma follows.
\end{proof}

Since each generator contains points of only a single color, we can
view the generators as being colored too. If we group the generators
by color, we observe that a colored storage scheme for points of
$\col$ colors is simply a collection of $\col$ uncolored storage
schemes. Similar to the uncolored case, let $t(P, \Gamma, Q)$ be the
cost of answering query $Q$ using storage scheme $\Gamma$ on the
colored point set $P$. The above lemma implies that
$t(P,\Gamma, Q) = \sum_c t(P_c,\Gamma_c, Q)$; the weights for each
color $c$ must be summed independently, using the generators
$\Gamma_c$ of color $c$, so the total query cost is simply the sum of
the query cost for all colors. Essentially, the problems are entirely
separate, and thus \cref{lem:chazellesLemma} applies to each color separately.

We first prove the following technical lemma, which we then use in our
lower bound.

\begin{lemma}
  \label{lem:probabilityHalfSuccesses}
  Consider a Bernoulli process with $n$ independent trials, each of
  which succeeds with probability $p$. Let $X$ be a random variable
  denoting the number of successes. If $p \geq 0.5$, then
  $\P(X \geq n/2) \geq p$.
\end{lemma}
\begin{proof}
We look at the first trial separately. Let $Y$ be a random variable denoting the number of successes in the remaining $n-1$ trials, and let $Y' = n - 1 - Y$ be the number of failures in the remaining trials. We wish to show that:

\begin{align}
\P(X \geq n /2) & \geq p  \\
p\P(Y \geq n/2 - 1) + (1-p) \P(Y \geq n/2) & \geq p \\
(1-p) \P(Y \geq n/2) & \geq p (1 - \P(Y \geq n/2 - 1)) \\
(1-p) \P(Y' \leq n/2 - 1) & \geq p \P(Y < n/2 - 1) \\
(1-p) \sum_{i = 0}^{\lfloor n/2 \rfloor - 1} \binom{n-1}{i} p^{n - i - 1} (1-p)^{i} & \geq p \sum_{i = 0}^{ \lceil n/2 \rceil - 2} \binom{n-1}{i} p^{i} (1-p)^{n - i - 1}\\ 
\sum_{i = 0}^{\lfloor n/2 \rfloor - 1} \binom{n-1}{i} p^{n - i - 1} (1-p)^{i + 1} & \geq \sum_{i = 0}^{ \lceil n/2 \rceil - 2} \binom{n-1}{i} p^{i + 1} (1-p)^{n - i - 1}
  \label{eq:final_equation}
\end{align}

Explanation of steps:
\begin{description}
    \item[(1) -- (2):] if the first trial succeeds (with probability $p$) we only need $n/2-1$ more successes in $Y$, otherwise (with probability $(1-p)$) we need $n/2$ more successes. 
    \item[(2) -- (3):] rewriting.
    \item[(3) -- (4):] the probabilities $\P(Y \geq n/2)$ and $\P(Y' \leq n/2 - 1)$ are equivalent (by definition of $Y'$), and the probabilities $\P(Y \geq n/2-1)$ and $\P(Y < n/2 - 1)$ are complementary and add up to $1$.
    \item[(4) -- (5):] summing over probability of getting exactly $i$ failures for all integers $0 \leq i \leq n/2 - 1$ on the left, and the probability of getting exactly $i$ successes for each integer $0 \leq i < n/2 - 1$ on the right.
    \item[(5) -- (6):] distributing the $(1-p)$ and $p$ factors over the summands.
\end{description}

The above equations are all equivalent. In order to show that \cref{eq:final_equation} holds we compare each summand separately. If we can show that for each $0 \leq i \leq \lceil n/2 \rceil -2$ the left hand side is larger than the right hand side, this certainly also holds for the whole sum; the left hand side even has one extra summand if $n$ is even, which we can ignore since it is non-negative. We thus wish to show that:

\begin{align}
\binom{n-1}{i} p^{n - 1 - i} (1-p)^{i + 1} & \geq \binom{n-1}{i} p^{i+1} (1-p)^{n - 1 - i} \\
 p^{n -2i - 2} & \geq (1-p)^{n -2i -2} \\
 p & \geq (1-p) 
\end{align}

\begin{description}
    \item[(7) -- (8):] rewriting and canceling of terms
    \item[(8) -- (9):] $n - 2i - 2$ is positive for $0 \leq i \leq \lceil n/2 \rceil -2$, so Equation (9) implies Equation (8).
\end{description}

Equation (9) follows from the assumption that $p \geq 0.5$. The lemma follows.
\end{proof}

We can now prove the main result of this section:

\begin{theorem}
  \label{lem:lower_bound}
  Let $\eps > 0$ be a constant, let $P$ be a random set of $n$
  points in $\R^d$ of $\col$ colors such that there are $n/\col$ points
  of each color, each with a weight from some faithful semigroup, let
  $\Gamma$ be any storage scheme for $P$ of size $m$, and let $q$ be a
  random point defining a dominance query $Q$. There exists a constant
  $R$ (depending on $\eps$) such that, with probability at least
  $(1-\eps)$, it holds that
  $t(P,\Gamma,Q) \geq R\col \left( \frac{\log (n / \col)}{\log (2m / n)} \right)
  ^{d-1}$.
\end{theorem}
\begin{proof}
Assume that $(1-\eps) \geq 0.5$; otherwise the result follows from using the lemma with $(1-\eps) = 0.5$. 

As we noted before, a storage scheme on $P$ consists of $\col$
separate storage schemes, one for each color $c$. For each of these
schemes $\Gamma_c$, \Cref{lem:chazellesLemma} holds (using the same
value for $\eps$). Let $q$ be a random query point defining a dominance query $Q$, and consider a
color $c$. Let the size of the storage scheme for color $c$ be $m_c$,
with $\sum_c m_c = m$. \Cref{lem:chazellesLemma} tells us that, with
probability $(1-\eps)$, $Q$ will be a ``difficult'' query for color $c$,
i.e.
$t(P_c,\Gamma,Q) \geq r_c \left( \frac{\log (n / \col)}{\log (2m_c/(n / \col))} \right) ^{d-1}$
(since each color contains $n/\col$ points). This holds for every color $c$. However,
$Q$ does not need to be ``difficult'' for every color; it is sufficient
if it is ``difficult'' for, say, half the colors, which by
\Cref{lem:probabilityHalfSuccesses} happens with probability at least
$(1-\eps)$. Let $\hat{C}$ be this set of difficult colors.

The cost to answer query $Q$ is thus
$t(P,\Gamma,Q) \geq \sum_{c \in \hat{C}} r_c\left( \frac{\log (n / \col)}{\log (2m_c/(n / \col))} \right) ^{d-1}$. Let $R' = \min_{c \in \hat{C}} r_c$
be the lowest among the constants for all difficult colors; then
$$t(P,\Gamma,Q) \geq \sum_{c \in \hat{C}} r_c \left( \frac{\log (n / \col)}{\log (2m_c/(n / \col))} \right) ^{d-1} \geq R'\sum_{c \in \hat{C}} \left( \frac{\log (n / \col)}{\log (2m_c/(n / \col))} \right) ^{d-1}$$. Assume for now that
the $m_c$'s are distributed evenly, that is, the number of generators
for each color is $m / \col$. That gives total query cost
$$t(P,\Gamma,Q) \geq R'\sum_{c \in \hat{C}} \left( \frac{\log (n / \col)}{\log (2(m / \col )/(n / \col))} \right) ^{d-1} \hspace{-1em}= R'\sum_{c \in \hat{C}} \left( \frac{\log (n / \col)}{\log (2m / n )} \right) ^{d-1} \hspace{-1em}\geq R'\col/2 \left( \frac{\log (n / \col)}{\log (2m / n )} \right) ^{d-1}$$. 
Choosing $R = R'/2$ gives the claimed result.

We now show that we did not lose generality because of the assumption
that the $m_c$'s are divided evenly. Since we want to prove a lower
bound, we want our statement to hold for any distribution of $m_c$'s,
and in particular also for the distribution for which $t(P,\Gamma,Q)$
is the smallest; we thus have to prove that dividing the $m_c$'s
evenly minimizes $t(P,\Gamma,Q)$. 

Let $L_c = \log (2m_c/(n / \col))$ and let $f(m_c) = \left( \frac{1}{L_c} \right) ^{d-1}$, and note that minimizing $\sum_c f(m_c)$ also minimizes $R' \sum_c \left( \frac{\log (n / \col)}{\log (2m_c/(n / \col))} \right) ^{d-1}$ \. The function $f(m_c)$ is convex w.r.t. $m_c$: the second derivative $f''(m_c) = \frac{(d-1) L_c^{-d-1} (d + L)}{m_c^2}$ is non-negative for any $m_c > \frac{n}{2\col}$ (the case where $m_c < \frac{n}{\phi}$ is irrelevant since we have $\frac{n}{\phi}$ points per color, and thus need at least $\frac{n}{\phi}$ generators per color). It follows that $\sum_c f(m_c)$ is convex too, and thus that choosing $m_c = m / \col$ for all colors $c$ indeed minimizes $\sum_c f(m_c)$.
\end{proof}

Since our bound is in the arithmetic model, it essentially gives a
lower bound on the number of values that we need to add together to
compute an answer. The real-valued pointer machine model is weaker than the
arithmetic model, so the lower bound directly holds there. In
contrast, in the word-RAM model we can encode multiple values in a
single word, and sum them up using a single operation, so the lower
bound does not hold there.

\section{A transformation for reducing space in 2D}
\label{sec:transformation}

For certain values of $\col$ and $s$ we can reduce the space usage in Lemma~\ref{lem:2d_dominance} for 2D dominance queries using one
level of ``bootstrapping''. The main observation is that using the
structure for weighted color frequency queries in $\R^3$, we can
actually support simultaneously answering a dominance query in $\R^1$
on a prefix $P_1,\dots,P_j$ of $t$ point sets in $\R^1$. This allows us
to reduce the space used at every node in the structure of
Lemma~\ref{lem:2d_dominance}.

\subparagraph{Querying multiple sets simultaneously.} Let $P_1,\dots,P_t$
be $t$ sets of points in $\R^1$, and let $n$ denote the total number
of points in $P=P_{\leq t} = \bigcup_{j \leq t} P_j$. Following the
same strategy as \cite{bozanis1995generalizedIntersectionSearching}, consider a point set $P_j$ and a color $c$, and let $x_i$ be the $i^\textrm{th}$ point of color $c$ in $P_j$. We map point $x_i$ to the point
$(x_i,x_{i+1},j) \in \R^3$, and associate it with weight $i$. We do this for every point set $P_j$ and every color $c$. Let $\tilde{P}$ be the resulting set of points. A quadrant of the form
$(-\infty, x] \times [y,\infty) \times (-\infty, i]$ is
\emph{$k$-shallow} with respect to a set $\tilde{P}$ if it contains at
most $k$ points of $\tilde{P}$. We then observe:

\begin{lemma}
  \label{obs:shallow}
  For any $1 \leq i \leq t$ and any $x,y$, the query
  $\tilde{Q} = (-\infty,x] \times [y,\infty) \times (-\infty,i]$ is
  $i\phi$-shallow with respect to $\tilde{P}$.
\end{lemma}

\begin{proof}
  For every $j \leq i$, $\tilde{Q}$ contains at most one point from
  $\tilde{P}_{j,c}$ for every color $c$; for $j > i$, $\tilde{Q}$ contains no points.
\end{proof}

Set $r=t\phi$ and build an $r$-shallow cutting for such quadrants on
$\tilde{P}_{\leq s}$ \citep{afshani18optim_deter_shall_cuttin_domin_ranges}. An
$r$-shallow cutting $\Xi$ is a set of $O(n/r)$ quadrants (of the above
form) such that (i) every $k\leq r$-shallow quadrant is contained in a
quadrant in $\Xi$, and (ii) every quadrant $\gamma \in \Xi$ contains
at most $O(r)$ points $\tilde{P}_\gamma$ from $\tilde{P}$. This set
$\tilde{P}_\gamma$ is the \emph{conflict list} of $\gamma$. One can
compute an $r$-shallow cutting of size $O(n/r)$ and process it so that
given a query point $q$ we can find a quadrant of $\Xi$ that contains
it (if such a quadrant exists) in $O(n\log n)$
time \citep{afshani18optim_deter_shall_cuttin_domin_ranges}.

We build a data structure for 3D weighted color frequency reporting
queries on the conflict list $\tilde{P}_\gamma$ for every quadrant
$\gamma \in \Xi$; say this uses $S(n)$ space for $n$ points, with $O(Q(n)+k)$ query time. Hence, the total space used is $O((n/r)S(r))$.

Now consider a query, consisting of an integer $j \in \{1,\dots,t\}$ and a
range $(-\infty, q]$. Observation~\ref{obs:shallow} guarantees that
the quadrant
$\tilde{Q}=(-\infty,q] \times [q,\infty) \times (-\infty,j]$ is
contained in a quadrant $\gamma \in \Xi$. Hence, querying with the
corner point $(q,q,j)$ we we find such a quadrant $\gamma$ in $O(\log
(n/r))$ time. We then query its associated data structure in $O(Q(r,k))$
time. Since $\gamma$ contains $\tilde{Q}$, any point contained in
$\tilde{Q}$ is contained in $\gamma$, and none of the points outside
of $\gamma$ can be contained in $\tilde{Q}$. Moreover, for each color
$c$; the total number of points in $P_\leq j$ is exactly the total
weight of the points in $\tilde{Q} \cap \tilde{P}_{\leq j, c}$. 

\begin{lemma}
  \label{lem:prefix_queries}
  Given an $S(n)$ space data structure that supports weighted color
  frequency queries for dominance ranges in $\R^3$ in $O(Q(n)+k)$
  time, we can build a data structure of size $O((n/t\phi)S(t\phi))$
  that stores $t$ sets of points $P_1,\dots,P_t$ in $\R^1$ of total size
  $n$. We can answer color frequency reporting queries for dominance
  ranges $(-\infty,q]$ in any $P_{\leq j} = \bigcup_{i \leq j} P_i$ in
  $O\left(\log \frac{n}{t\phi} + Q(t\phi) + k \right)$ time.
\end{lemma}

Note that unfortunately we cannot easily extend this idea to
simultaneously answer dominance queries on a prefix $P_1,\dots,P_j$ of
points in $\R^2$, as shallow cuttings in $\R^4$ use quadratic space.

\subparagraph{Reducing the space for dominance queries in $\R^2$.} For
points in $\R^2$ and parameter $s$, we now follow the same recursive
approach as in Theorem~\ref{lem:2d_dominance},
recursively partitioning the set of points using $t=s$ vertical strips
$S_1,\dots,S_s$. However, we now build the data structure from
\Cref{lem:prefix_queries} on (the $y$-coordinates of) the point sets $P_j = S_j \cap P$, instead
of constructing $s$ one-dimensional structures on (the $y$-coordinates of
the) points in each $L_j = \bigcup_{i<j} P_i$. At every level we thus
use $O\left(\frac{S(s\phi)n}{s\phi}\right)$ space, and thus
$O\left(\frac{S(s\phi)n\log_s n}{s\phi}\right)$ space in total.

Consider a dominance query whose corner point lies in strip $S_j$. We
can now compute the contribution of the points in strips
$S_1,\dots,S_{j-1}$ in $O(\log(n/(s\phi)) + Q(s\phi) + k)$ time by
querying the Lemma~\ref{lem:prefix_queries} data structure. Therefore, by doing this on all $O(\log_s n)$ levels of the tree, 
the total query time becomes
$O(\log n + \log_s n (k+\log(n/(s\phi)) + Q(s\phi)))$.

\begin{lemma}
  \label{lem:shallow_data_structure}
  Let $P$ be a set of $n$ points in $\R^2$, each of which has a color from $\{1 \dots \col\}$, and let $2 \leq s \leq n$ be a parameter. Given an
  $S(n)$ space data structure that supports weighted color frequency
  queries for dominance ranges in $\R^3$ in $O(Q(n)+k)$ time, then we
  can build a data structure of size
  $O\left(\frac{S(s\phi)n\log_s n}{s\phi}\right)$ that can answer a
  color frequency reporting query for a dominance range in
  $O(\log n + \log_s n(k+\log(n/(s\phi)) + Q(s\phi)))$ time.
\end{lemma}

In particular, starting with the data structure from
Theorem~\ref{thm:tradeoff_dominance_data_structure} (set to use
$O(r^{1+\eps})$ space for $r$ points) gives an $O(ns^\eps\phi^\eps \log_s n)$
space structure with query time $O(\log_s n(\log n + k))$. Hence, it
allows us to reduce the space by a factor $s^{1-\eps} / \col^\eps$ while not losing anything in the query time when
$k$ and $\col$ and $s$ are sufficiently large.

Consider for example the case when $\col = \Theta(\log n)$, for a query where $k = \Theta(\log n)$, and we choose $s = \log n$. \Cref{lem:2d_dominance} yields an $O(n s \log_s n) = O\left(n \frac{\log^2 n}{\log \log n} \right)$ space data structure with $O(\log n + k \log_s n) = O\left(\frac{\log^2 n}{\log \log n} \right)$ query time, and \Cref{lem:shallow_data_structure} yields an $O(n (s \col)^{\eps} \log_s n) = O\left(n \frac{\log^{1+ 2\eps} n}{\log \log n} \right)$ space data structure with $O(\log_s n (\log n + k)) = O\left(\frac{\log^2 n}{\log \log n} \right)$ query time; the latter improves the space usage by almost a factor $\log n$.

\begin{colorallary}
\label{col:shallow_datastruct_specific_case}
When $\col = \Theta(\log n)$, we can build an $O(n \log^{2\eps} n \frac{\log n}{\log \log n})$ space data structure with $O(\frac{\log^2 n}{\log \log n})$ query time.
\end{colorallary}

\section{Low-space algorithms in 2D}
\label{sec:A_linear_space_algorithm}

In this section we turn to the algorithmic problem in which we are
given a set $\Q$ of $m$ queries that we all wish to answer. For every
query $Q \in \Q$ we must output a list of $k_Q$ (color,count) tuples,
one for each color contained in $Q$. The colors in each of these lists do not have to be in
any specific order, as long as all color counts of a specific query
are together. Ideally, we provide these answers in a stream so that we
do not have to store the entire output of size $K = \sum_Q k_Q$.

Note that for point sets in $\R^1$ we can simply build and query the
linear space solution from \cite{bozanis1995generalizedIntersectionSearching}. For point
sets in $\R^2$, one approach could be to first perform a color
reporting query \citep{shi2005optimalColoredAlgs} in order to
construct, for each color $c$ with $n_c$ points $P_c$, the set of
$m_c$ queries $\Q_c$ that contain at least one point of color $c$;
then $\sum_c n_c = n$ and $\sum_c m_c = K$. We can now use the algorithm by \cite{chan2010countingInversionsOfflineCounting} for offline uncolored orthogonal range counting, which is in the word-RAM model. For each color $c$ we separately compute the count
for color $c$ for each query in $\Q_c$ in
$O((n_c + m_c) \sqrt{\log (n_c + m_c)})$ time, resulting in total time
$O(\sum_c (n_c + m_c) \sqrt{\log (n_c + m_c)}) \leq O(\sqrt{\log (n +
  m)} \sum_c (n_c + m_c)) = O((n + K) \sqrt{\log (n + m)})$. Note that,
since we require all the color counts of a single query to be grouped,
we would need to store the entire output in memory using $O(K)$
additional working space.

In the remainder of this section we focus on dominance queries. A simple approach would be to simply build the data structure from
\Cref{thm:tradeoff_dominance_data_structure} and use it to answer
each query, using $O(n (s \log_s n)^{d-1})$ working space. We now show how to reduce the working space of
this algorithm by a factor $s \log_s n$, while obtaining the same
runtime. This thus yields a linear space algorithm in $\R^2$.

\begin{theorem}
  \label{lem:dominanceOfflineSweep}
  Let $P$ be a set of $n$ points in $\R^d$, with $d \geq 2$, each of
  which has a color from $\{1, \dots, \col\}$, and let $\Q$ be a set of $m$
  dominance queries, and let $2 \leq s \leq n$ be a parameter. We can
  answer all queries in
  $O(n \log n (s \log_s n)^{d-1} + m \log n \log_s^{d-2} n +
K \log_s^{d-1} n)$ time, where
  $K$ is the total output size, using $O(n(s \log_s (n))^{d-2} + m)$
  working space.
\end{theorem}
\begin{proof}
We follow a similar approach as in \Cref{thm:tradeoff_dominance_data_structure}. We use induction on $d$, and for the inductive step we build a strip-tree with fanout $s$ on the $x$-coordinates of the points. We do not build the $(d-1)$-dimensional data structures yet, as they would take up too much space. The main idea is to build them only when we need them, and destroy them afterwards. We additionally distribute the queries $\Q$ over the tree, placing each query $Q$ in the leaf strip containing its cornerpoint $q$.

We sweep a plane $z$, that is perpendicular to the $x$-axis, from $x = -\infty$ towards $x = \infty$. Consider the root node of the recursion tree. Whenever $z$ enters a strip $S_i$ of this root node, we build the $(d-1)$-dimensional data structure $\L_i$ on $L_i'$. When $z$ leaves strip $S_i$ we destroy this data structure again (and subsequently build the data structure $\L_{i+1}$ when $z$ enters strip $S_{i+1}$). We recursively do the same for the recursive $d$-dimensional data structure on strip $S_i$. As such, each point is contained in at most one $(d-1)$-dimensional data structure at a time, so the total space required for these data structures at any time is $O(n(s \log_s n)^{d-2})$. Including the space required for storing the queries in the tree this yields $O(n(s \log_s n)^{d-2} + m)$ working space.

Whenever $z$ enters a leaf containing a query $q$ we immediately answer this query using the same procedure as before. This is possible since the $(d-1)$-dimensional data structures on $L_i'$ for every strip $S_i$ containing $z$ are exactly the data structures we wish to query.

We build every $(d-1)$-dimensional data structure exactly once, so the total runtime is still the sum of building the data structure from \Cref{thm:tradeoff_dominance_data_structure} ($O(n \log n (s \log_s n)^{d-1})$) plus the cost of querying it $m$ times ($O(m \log n \log_s^{d-2} n +
K \log_s^{d-1})$). The sweeping adds only $O(n \log n)$ time to this, which is dominated.
\end{proof}

We would like to reduce the space usage even further, preferably all the way down to $O(n + m)$ for
any dimension $d$, but this may not be as easy as it sounds.

One idea would be to apply the sweeping strategy from \Cref{lem:dominanceOfflineSweep} across multiple
dimensions, not just one. Consider a $(d-1)$-dimensional data
structure built on some set of points $L_i'$ whenever $z$ entered a
strip $S_i$. The queries that are asked of this data structure come in
order of increasing $x$-coordinate, but in arbitrary order
w.r.t. their other coordinates; this means we can not easily use a
sweeping strategy in more dimensions.

Another idea is to use the fact that we are given all queries in
advance, and rather than building a $(d-1)$-dimensional data structure
on $L_i'$ and querying it one by one, instead performing all queries
on $L_i'$ at once. However, this makes combining the answers of all these queries more difficult. The technique
discussed in~\Cref{thm:tradeoff_dominance_data_structure} for combining the answers of
multiple sub-queries only works if we are handling one query at a time. If
we obtain a partial answer for all $m$ queries, each of which may have
to be combined with other partial answers, we have two issues. Firstly we would need $O(K)$
additional storage to store a table with all these partial answers,
and secondly we would need additional time to find the correct entry in
the table for each partial answer, resulting in e.g. an $O(K \log \col)$ term in the running time.

For the same reason, it is difficult to extend this result to general more-sided queries using techniques similar to \Cref{lem:data_structureIncreaseNumberOfSides}; this would also require combining many sub-queries. We can however increase the number of sides by $1$, using the fact that we can answer dominance queries in order of any coordinate.

\begin{theorem}
Let $P$ be a set of $n$ points in $\R^2$, each of which has a color from $\{1,\dots,\col\}$, and let $\Q$ be a set of $3$-sided queries of the form $[x_1,x_2] \times (-\infty,y)$, and let $2 \leq s \leq n$ be a parameter. We can answer all queries in $O(n \log^2 n s \log_s n + m \log n + K \log_s n)$ time, using $O(n + m)$ space.
\end{theorem}
\begin{proof}
We use the same technique as in \Cref{lem:data_structureIncreaseNumberOfSides} by building a binary tree $T$ on the $x$-coordinates of $P$. Additionally we distribute all queries over the tree, by placing each query $Q = [x_1,x_2] \times -\infty,y]$ in the highest node $u \in T$ such that $x_u \in [x_1, x_2]$. For a node $u$ with children $c_1$ and $c_2$ let $\Q_u$ be the set of queries placed in it, with total output complexity $K_u$. We split each $3$-sided query $Q \in \Q_u$ in two dominance queries $Q_1$ and $Q_2$ as before, resulting in the sets $\Q_1$ and $\Q_2$. Importantly, note that each pair $Q_1$ and $Q_2$ have the same $y$-coordinate. We use \Cref{lem:dominanceOfflineSweep} twice to solve the queries $\Q_1$ on the points $P_{c_1}$, and the queries $\Q_2$ on the points $P_{c_2}$. We are then left with two lists of query answers that we need to combine. We observe that \Cref{lem:dominanceOfflineSweep} yields the answers to the queries in order of the sweep direction; by sweeping in the $y$-direction, we can thus receive the answers to queries $\Q_1$ and $\Q_2$ in order of $y$-coordinate. Both lists of answers will thus be in the same order, and we can combine them in $O(K_u)$ time by simply scanning through both lists simultaneously.

At each node $u$ with $n_u$ points and $m_u$ queries with output complexity $K_u$ we thus use \linebreak$O(n_u \log n_u s \log_s n_u+ m_u \log n_u n_u + K_u \log_s n_u)$ time. Since each query appears only once in the tree, and each point appears $O(\log n)$ times, this results in $O(n \log^2 n s \log_s n + m \log n + K \log_s n)$ total time.
\end{proof}

\section{Concluding remarks}
In this article we have considered the frequency reporting problem. We have presented a simple parameterized data structure, and have shown an almost-matching lowerbound for the weighted case in the arithmetic model. A prominent open question is whether we can still actually achieve strictly output sensitive $O(f(n) + k)$ query time with linear space by circumventing the lower bound; for example in the case where all weights are one, so the frequency reporting case; or by exploiting the fact that we can subtract counts; or in a different computational model. Note that even in the uncolored case, where the lower bound for counting queries uses the same assumptions, no such results are known.

Here, we focused on axis-aligned query boxes. It would be interesting to know if the parameterized data structure could be adapted to more general or different query ranges, such as (regular) convex polygons or circles. 

For the offline problem in $\R^2$, we presented a linear space algorithm that was strictly output sensitive, but which essentially used $O(n^\eps + k)$ time per query for $m = \Theta(n)$ queries. Another open question is whether it is possible to reduce this $n^\eps$ factor, ideally all the way to $\log n$, so using $O((n + m) \log n + K)$ time to answer all $m$ queries using $O(n + m)$ space. Our algorithm currently uses linear space in $\R^2$, but super-linear space for $d > 2$; new techniques could also help in achieving a strictly output sensitive algorithm using linear space in higher dimensions.

\bibliographystyle{abbrvnat}
\bibliography{references}

\end{document}